\documentclass[letterpaper, 10 pt, conference]{ieeeconf}  
\makeatletter
\newcommand*\titleheader[1]{\gdef\@titleheader{#1}}
\AtBeginDocument{%
  \let\st@red@title\@title
  \def\@title{%
    \bgroup\normalfont\footnotesize\vskip-0.9in\@titleheader\par\egroup
    \vskip0.6in\st@red@title}
}
\makeatother

\IEEEoverridecommandlockouts                              
\usepackage{cite}
\usepackage{url}
\let\labelindent\relax
\usepackage{enumitem}
\usepackage{amsmath}
\usepackage{amsfonts}

\usepackage{amsthm}
\usepackage{amssymb}
\usepackage{mathtools}
\usepackage{dsfont}
\usepackage{graphicx}
\usepackage[font=it]{caption}
\usepackage{subcaption}
\usepackage{xcolor}
\usepackage{algorithm, algorithmic}
\usepackage{comment}
\usepackage[english]{babel}
\usepackage[utf8]{inputenc}
\usepackage{fancyhdr}
\usepackage{textcomp}

\newtheorem{theorem}{Theorem}[section]

\newtheorem{lemma}{Lemma}[section]

\newtheorem{remark}{Remark}[section]
\newtheorem{assumption}{Assumption}[section]

\DeclareMathOperator*{\argmin}{arg\,min}

\renewcommand{\Re}{\mathbb{R}}

\title{\LARGE \bf
Hybrid Heavy-Ball Systems:\\Reset Methods for Optimization with Uncertainty
}

\titleheader{An abridged version of this article will be published in the Proceedings of the 2021 American Control Conference (\textcopyright 2021 AACC).}

\author{Justin H. Le and Andrew R. Teel
\thanks{}
\thanks{Justin H. Le and Andrew R. Teel are with the ECE Department, Univer-
sity of California, Santa Barbara, CA 93106-9560. Email: teel@ucsb.edu, justinle@ucsb.edu. This work has been supported by the Air Force Office of Scientific Research through the grant FA9550-18-1-0246.}
}

\begin{document}

\maketitle
\thispagestyle{empty}
\pagestyle{empty}

\begin{abstract}
Momentum methods for convex optimization often rely on precise choices of algorithmic parameters, based on knowledge of problem parameters, in order to achieve fast convergence, as well as to prevent oscillations that could severely restrict applications of these algorithms to cyber-physical systems. To address these issues, we propose two dynamical systems, named the Hybrid Heavy-Ball System and Hybrid-inspired Heavy-Ball System, which employ a feedback mechanism for driving the momentum state toward zero whenever it points in undesired directions. We describe the relationship between the proposed systems and their discrete-time counterparts, deriving conditions based on linear matrix inequalities for ensuring exponential rates in both continuous time and discrete time. We provide numerical LMI results to illustrate the effects of our reset mechanisms on convergence rates in a setting that simulates uncertainty of problem parameters. Finally, we numerically demonstrate the efficiency and avoidance of oscillations of the proposed systems when solving both strongly convex and non-strongly convex problems.
\end{abstract}


\section{Introduction}
\label{sec:intro}

Convex optimization problems are becoming increasingly challenging as they find broader applications in cyber-physical systems, where they often bring stringent requirements on the efficiency and robustness of the algorithms that are used to solve them. In theory and in practice, the efficient convergence of iterative algorithms for convex optimization can be achieved through the use of momentum, as in the sense of Nesterov's method \cite[Ch. 2]{nesterov-2014}. In particular, for certain strongly convex problems, one theoretically significant aspect of Nesterov's method is that it achieves an efficient degradation of its convergence rate as the conditioning of the problem tends to infinity \cite[Sec. 4.5]{lessard-2016}. However, in order to maintain such a property, the algorithmic parameters, namely the stepsize and momentum parameter, must be selected according to a specific formula dependent on the problem parameters. When such parameters are unavailable, Nesterov's method, as well as other momentum methods, not only lose their theoretical guarantees of efficiency but also suffer from oscillations in their trajectories that hinder their convergence in practice \cite{o'donoghue-2015}. Moreover, such oscillations can make momentum methods unreliable for applications in feedback-based optimization \cite{hauswirth-2020} \cite{ortmann-2020}, where the algorithm is used in feedback interconnection with a physical system and can thereby jeopardize the safety of that system when experiencing oscillations.

Oscillations in momentum methods have been addressed with reset mechanisms \cite{fercoq-2019} \cite{roulet-2017} \cite{fercoq-2019b}, many of which involve scheduling the times at which the resets occur without using any feedback information about the state of the algorithm but instead by using knowledge or estimates of certain problem parameters, which are often uncertain or entirely unknown in practice. In contrast, the adaptive reset mechanisms of \cite{o'donoghue-2015} offer conditions that can be computed straightforwardly at each iteration using state information to determine the instants of reset. Although the theoretical guarantees of efficiency under the adaptive mechanism of \cite{o'donoghue-2015} are difficult to extend beyond quadratic objectives, the proposed reset conditions raise analogies with reset systems in control theory, especially those that have benefited from the use of hybrid systems theory \cite{prieur-2018}, suggesting opportunities to analyze and design novel optimization algorithms with resets within a hybrid systems framework, as done in \cite{teel-2019}.

Our work builds on the themes of \cite{teel-2019} and \cite{prieur-2018} in order to determine whether or not (and to what extent) a feedback-based reset mechanism can either improve or degrade the efficiency and robustness to uncertainty of momentum methods. First, we introduce a hybrid dynamical system, referred to as the Hybrid Heavy-Ball Method (HHBM), that incorporates momentum in its flows and uses an adaptive mechanism to reset the momentum to zero whenever it points away from the negative gradient of the objective function. We also introduce a differential inclusion, referred to as the Hybrid-inspired Heavy-Ball Method (HiHBM), that uses a similar mechanism to adjust the amount of damping of the momentum. These two systems serve as vehicles for investigating the effects of reset mechanisms in existing momentum methods and for deriving novel momentum methods that are useful for their robustness to uncertainty. Toward these goals, we first derive linear matrix inequality (LMI) conditions for HHBM to achieve exponential convergence in the continuous-time sense for the case of strongly convex quadratic objectives, in order to relate our proposed ideas to an existing result on linear reset systems \cite{nesic-2008}. Using insights gained from the proof, we formulate analogous LMI conditions, assuming only strong convexity, for a general class of discrete-time systems whose special cases include discretizations of HHBM, HiHBM, and several other dynamical systems of interest in optimization, including the Hybrid Hamiltonian Algorithm of \cite{teel-2019}. Our discrete-time analysis generalizes known LMI conditions in the literature on momentum methods, extending prior results from time-invariant systems to systems that feature switching behaviors based on state-feedback information.

We show numerically that our LMI conditions can be used to compute feasible exponential rates for the proposed family of discrete-time systems, revealing that the proposed reset laws derived from HHBM and HiHBM can mitigate the deterioration of rates caused by uncertainty about problem parameters. The computations suggest that the reset laws do not quite preserve the rate guarantees of Nesterov's method when assuming perfect knowledge of problem parameters; on the other hand, we demonstrate that the discrete-time analogues of HHBM and HiHBM show promise in achieving fast convergence and reduction of oscillations without requiring precise tuning of algorithmic parameters, in contrast to Nesterov's method. Finally, we compare the performance of the considered methods for a non-strongly convex objective to show that our proposed algorithms exhibit similar advantages over existing methods as in the strongly convex case, even when existing methods are tuned extensively by experiment.

\section{Problem Setting and Algorithmic Overview}
Consider the problem
\begin{align}
\label{problem:optimization}
    \min_{q \in \Re^n} \; \phi(q)
\end{align}
under the following assumption:
\begin{assumption}
\label{objective_assumptions}
The objective $\phi: \Re^n \to \Re$ 
\begin{enumerate}
    \item attains $\phi^* \coloneqq \min_{q \in \mathbb R^n} \phi(q) > -\infty$,
    \item is continuously differentiable,
    \item has compact sublevel sets,
    \item has an $L$-Lipschitz\footnote{A function is said to be $L$-Lipschitz if it is Lipschitz continuous with Lipschitz constant $L$.} gradient $\nabla\phi$,
    \item is invex \cite{ben-israel-1986}, i.e., satisfies $$\mathcal{Q}^* \coloneqq \{q \in \Re^{n}: \phi(q) = \phi^*\} = \{q \in \Re^{n}: \nabla\phi(q) = 0\}.$$
\end{enumerate}
\end{assumption}

For the purpose of computing a solution to \eqref{problem:optimization}, we propose algorithms that are extensions of the following system, referred to as a heavy-ball system with parameter $K \in \Re_{>0}$, denoted HB($K$) and with state denoted $x \coloneqq (q, p)$:
\begin{align}
\label{eq:hb}
    \dot x &= \left[\begin{array}{cc}
        p \\
        -Kp - \nabla\phi(q)
    \end{array}\right].
\end{align}

We propose the Hybrid Heavy-Ball Method (HHBM) with parameter $K \in \Re_{>0}$, denoted HHB($K$), which is a hybrid system \cite[Ch. 2]{goebel-2012} with state $z \coloneqq (x, \tau)$, with $x \coloneqq (q, p)$, with flow map and jump map
\begin{subequations}
\label{eq:hhb}
\begin{align}
\label{eq:map_hhb}
    \begin{array}{ll}
        \dot x = f_0(x) 
    \coloneqq \left[\begin{array}{cc}
        p  \\
        -Kp - \nabla\phi(q) 
    \end{array}\right], & \dot\tau = 1, \\
        x^+ = g_0(x) \coloneqq \left[\begin{array}{cc}
        q  \\
        0 
    \end{array}\right], & \tau^+ = 0,
    \end{array}
\end{align}
with flow set $\mathcal{F}$ given by
\begin{align}
    \mathcal{F}_0 &\coloneqq \{(q, p) \in \Re^{2n}: \; \langle \nabla\phi(q), p\rangle \leq 0\}, \nonumber \\
    \mathcal{F} &\coloneqq \left(\Re^{2n} \times [0, \underline{T}]\right) \cup \left(\mathcal{F}_0 \times [\underline{T}, \infty)\right), \quad \underline{T} \in \Re_{>0}, \label{eq:flow_set_hhb}
\end{align}
and jump set $\mathcal{J}$ given by
\begin{align}
    \mathcal{J}_0 &\coloneqq \{(q, p) \in \mathbb R^{2n}: \; \langle \nabla\phi(q), p\rangle \geq 0\}, \nonumber \\
    \mathcal{J} &\coloneqq \mathcal{J}_0 \times [\underline{T}, \infty). \label{eq:jump_set_hhb}
\end{align}
\end{subequations}
The parameter $\underline{T}$ provides temporal regularization in the sense of \cite{nesic-2008} to avoid purely discrete-time solutions. Without the regularization, the fact that $\mathcal{F}_0 \cap \mathcal{J}_0 \neq \emptyset$ would allow for solutions \cite[Def 2.6]{goebel-2012} that jump indefinitely. Note that, for sufficiently small $\underline{T}$, the special case HHB($0$) is closely related to the Hybrid Hamiltonian Algorithm \cite{teel-2019}.


We also propose a differential inclusion referred to as the Hybrid-inspired Heavy-Ball Method (HiHBM) with parameters $\{\underline{K}, \overline{K}\} \in \Re^2$ satisfying $0 < \underline{K} \leq \overline{K}$, denoted HiHB($\underline{K}$, $\overline{K}$), with state $x \coloneqq (q, p)$ and dynamics given by
\begin{subequations}
\label{eq:hihb}
\begin{align}
    \dot x &\in F(x) \coloneqq \left[\begin{array}{cc}
        p \\
        -\kappa(x)p - \nabla\phi(q)
    \end{array}\right], \label{eq:map_hihb} \\
    \kappa(x) &\coloneqq \kappa(x; \underline{K}, \overline{K}) \nonumber \\
    &\coloneqq \begin{dcases}
        \overline{K} & \text{if } \langle\nabla\phi(q), p\rangle > 0, \\
        \underline{K} & \text{if } \langle\nabla\phi(q), p\rangle < 0, \\
        {\left[\underline{K}, \overline{K}\right]} & \text{if } \langle\nabla\phi(q), p\rangle = 0.
    \end{dcases} \label{eq:switch}
\end{align}
\end{subequations}
For any $K \in \Re_{>0}$, HiHB($K$, $K$) is equivalent to HB($K$).

Under Assumption \eqref{objective_assumptions}, it can be shown that $\mathcal{Q}^*$ is uniformly globally asymptotically stable \cite[Def. 3.6]{goebel-2012} for the system \eqref{eq:hhb}, with the proof being similar to that of \cite[Thm. 1]{teel-2019}. An analogous result can be obtained for \eqref{eq:hihb}. However, we omit the details due to a lack of space, considering that the current work instead focuses on studying exponential rates of convergence.

\section{Continuous-Time Exponential Rates}
A differentiable function $\phi$ is said to be $\mu$-strongly convex if, for some $\mu \in \Re_{>0}$, it holds that,  for all $x, y \in \Re^{n}$,
\begin{align*}
    \phi(y) \geq \phi(x) + \nabla\phi(x)^T (y - x) + \frac{\mu}{2}|y - x|^2.
\end{align*}
We use the following property of strongly convex functions with Lipschitz gradient, which is established in \cite[Eq. 3.27]{fazlyab-2018} as a special case of \cite[Lemma 6]{lessard-2016}.
\begin{lemma}
\label{thm:iqc}
Let $\phi: \Re^{n} \mapsto \Re$ be $\mu$-strongly convex with $L$-Lipschitz gradient. Defining
\begin{align}
\label{eq:iqc_matrix}
    M_{\mu, L} \coloneqq \left[\begin{array}{cc}
        -\frac{\mu L}{\mu + L}I_n & \frac{1}{2}I_n \\
        \frac{1}{2}I_n & -\frac{1}{\mu + L}I_n
    \end{array}\right],
\end{align}
it holds that, for all $v, w \in \Re^n$,
\begin{align}
\label{eq:iqc}
\left[\begin{array}{c}
        v - w \\
        \nabla\phi(v) - \nabla\phi(w)
    \end{array}\right]^T
M_{\mu, L}
\left[\begin{array}{c}
        v - w \\
        \nabla\phi(v) - \nabla\phi(w)
    \end{array}\right] \geq 0.
\end{align}
\end{lemma}

We now establish a result regarding exponential convergence of HHBM.
\begin{theorem}
\label{thm:lmi_hhbm}
For $K \in \Re_{> 0}$, let
\begin{equation}
\begin{aligned}
\label{eq:iqc_plant}
    A &\coloneqq \left[\begin{array}{cc}
        0 & I_n \\
        0 & -K I_n
    \end{array}\right], &
    A_R &\coloneqq \left[\begin{array}{cc}
        I_n & 0 \\
        0 & 0
    \end{array}\right], \\
    B &\coloneqq \left[\begin{array}{c}
        0 \\
        -I_n
    \end{array}\right], &
    C &\coloneqq \left[\begin{array}{c}
        I_n \\
        0
    \end{array}\right]^T.
\end{aligned}    
\end{equation}
Consider a $\mu$-strongly convex quadratic function $\phi$ satisfying Assumption~\eqref{objective_assumptions} with unique global minimum denoted $q^*$, and let $x^* \coloneqq (q^*, 0)$. Defining $M_{\mu, L}$ by \eqref{eq:iqc_matrix}, suppose that there exist $\alpha, \varepsilon, \sigma_{\phi}, \sigma_1, \sigma_2 \in \Re_{>0}$ and a positive definite $P \in \Re^{2n \times 2n}$ such that the matrices
\begin{align}
    M_\mathcal{F} &\coloneqq \left[\begin{array}{cc}
        PA + A^T P + 2\alpha P & PB \\
        B^T P & 0
    \end{array}\right], \nonumber \\
    M_\mathcal{J} &\coloneqq \left[\begin{array}{cc}
        A_R^T P A_R - P & 0 \\
        0 & 0
    \end{array}\right], \nonumber \\    
    M_{\phi} &\coloneqq \left[\begin{array}{cc}
        C^T & 0 \\
        0 & I_n
    \end{array}\right]
    M_{\mu, L}
    \left[\begin{array}{cc}
        C & 0 \\
        0 & I_n
    \end{array}\right], \label{eq:M_phi} \\
    M_{\varepsilon} &\coloneqq \left[\begin{array}{ccc}
        \varepsilon I_n & 0 & 0 \\ [\smallskipamount]
        0 & \varepsilon I_n & -\frac{1}{2}I_n \\ [\smallskipamount]
        0 & -\frac{1}{2}I_n & 0
    \end{array}\right], \quad M_0 \coloneqq M_{(\varepsilon=0)} \nonumber
\end{align}
satisfy
\begin{subequations}
\label{eq:lmi_hhb}
\begin{align}
    M_\mathcal{F} + \sigma_{\phi} M_{\phi} + \sigma_1 M_{\varepsilon} &\leq 0, \label{eq:lmi_hhb_flow} \\
    M_\mathcal{J} - \sigma_2 M_{0} &\leq 0. \label{eq:lmi_hhb_jump}
\end{align}
\end{subequations}
Then, there exist $\underline{T}^*, c \in \Re_{>0}$ such that, for all $\underline{T} \in (0, \underline{T}^*)$, the solutions of HHB($K$) satisfy
\begin{align}
    |x(t, j) - x^*| \leq c|x_0 - x^*|\exp\left(-\frac{\alpha}{\textnormal{cond}(P)}t\right)
\end{align}
for each initial condition $x(0, 0) \coloneqq x_0 \in \mathcal{F} \cup \mathcal{J}$ and for all $(t, j) \in \textnormal{dom }x$. Here, $\text{cond}(P) \coloneqq \lambda_{\text{max}}(P)/\lambda_{\text{min}}(P)$, where $\lambda_{\text{max}}(P)$ and $\lambda_{\text{min}}(P)$ denote the largest and smallest eigenvalues of $P$, respectively.
\end{theorem}
\begin{proof}
Define
\begin{align*}
    \tilde{x} &\coloneqq x - x^*, & \tilde{q} &\coloneqq q - q^*, \\
    e &\coloneqq (\tilde{x}, u), & u &\coloneqq \nabla\phi(q).
\end{align*}
We aim to apply \cite[Thm. 2]{nesic-2008} to the system whose state is $\tilde{x}$. To do so, we first observe that, under the given assumptions on $\phi$, there exists a positive definite matrix $\tilde{Q}$ such that $\nabla\phi(q) = \tilde{Q}(q - q^*) = \tilde{Q}\tilde{q}$, so that the flow map and jump map for the system with state $\tilde{x}$ can be written as
\begin{subequations}
\begin{align}
    \tilde{A} &\coloneqq \left[\begin{array}{cc}
        0 & I_n \\
        -\tilde{Q} & -K I_n
    \end{array}\right] \label{eq:closed_loop_matrix}, \\
    \tilde{f}_0(\tilde{x}) &= \tilde{A}\tilde{x}, \label{eq:flow_map_tilde} \\
    \tilde{g}_0(\tilde{x}) &= A_R \tilde{x}, \label{eq:jump_map_tilde}
\end{align}
\end{subequations}
with $A_R$ given by \eqref{eq:iqc_plant}. Then, letting $\tilde\phi(\tilde{q}) \coloneqq (1/2)\left|\tilde{Q}^{1/2}\tilde{q}\right|^2$, the flow set \eqref{eq:flow_set_hhb} and jump set \eqref{eq:jump_set_hhb} can be expressed in terms of $\tilde{x} \coloneqq (\tilde{q}, p)$ by observing that $(q, p) = (\tilde{q} + q^*, p)$ and that
\begin{align}
\label{eq:equal_gradients}
    \nabla_{\tilde{q}}\tilde\phi(\tilde{q}) = \nabla_q\phi(q), \quad \forall q \in \Re^n.
\end{align}
For brevity, we omit the subscript of $\nabla$ herein.

First, note that \cite[Assumption 1]{nesic-2008} holds because $g_0$ in \eqref{eq:jump_map_tilde} satisfies $g_0(\tilde{x}) = A_R \tilde{x} \in \mathcal{F}$. To satisfy the condition \cite[Eq. 20]{nesic-2008}, take $V(\tilde{x}) \coloneqq \tilde{x}^T P \tilde{x}$. Then, to show that \cite[Eq. 22]{nesic-2008} is satisfied, multiply \eqref{eq:lmi_hhb_jump} by $e^T$ and $e$ on the left and right, respectively, and observe that
\begin{align}
    e^T M_0 e = -p^T \nabla\phi(q) \leq 0, \quad \forall (q, p) \in \mathcal{J}. \nonumber
\end{align}
Then, \eqref{eq:equal_gradients} ensures that \cite[Eq. 22]{nesic-2008} holds for all points in the jump set of the system whose state is $\tilde{x}$.

Next, to show that \cite[Eq. 21]{nesic-2008} is satisfied, define
\begin{align}
\label{eq:inflation}
    \mathcal{F}_{\varepsilon} &\coloneqq \{(\tilde{q}, p) \in \Re^{2n}: \; p^T\nabla\tilde\phi(\tilde{q}) - \varepsilon\left(|\tilde{q}|^2 + |p|^2\right) \leq 0\}.
\end{align}
Then, multiply \eqref{eq:lmi_hhb} by $e^T$ and $e$ on the left and right, respectively. From \eqref{eq:equal_gradients}, we have the inequality
\begin{align}
    e^T M_{\varepsilon} e = -p^T \nabla\tilde\phi(\tilde{q}) + \varepsilon |\tilde{x}|^2 \geq 0, \quad \forall \tilde{x} \in \mathcal{F}_{\varepsilon}, \nonumber
\end{align}
which implies that
\begin{align}
    e^T \left(M_\mathcal{F}  + \sigma_{\phi} M_{\phi}\right) e \leq 0, \quad \forall \tilde{x} \in \mathcal{F}_{\varepsilon}. \nonumber
\end{align}
Then, observe that $e^T M_{\phi} e$ is equal to the left-hand side of \eqref{eq:iqc} when $v = q$ and $w = q^*$. Hence, Lemma~\ref{thm:iqc} can be applied to obtain $e^T M_{\phi} e \geq 0$. It follows from the previous inequality that 
\begin{align}
    e^T M_\mathcal{F} e \leq 0, \quad \forall \tilde{x} \in \mathcal{F}_{\varepsilon}. \nonumber
\end{align}
Expanding the left-hand side yields
\begin{align}
    \frac{\partial V}{\partial\tilde{x}} \tilde{A}\tilde{x} \leq -2\alpha V(\tilde{x}) \leq -2\alpha\lambda_{\text{min}}(P)|\tilde{x}|^2, \quad \forall \tilde{x} \in \mathcal{F}_{\varepsilon}. \nonumber
\end{align}
Furthermore, the condition \cite[Eq. 20]{nesic-2008} is satisfied according to $V(\tilde{x}) \leq \lambda_{\text{max}}(P)|\tilde{x}|^2$. Thus, letting $a_2 \coloneqq \lambda_{\text{max}}(P)$ and $a_3 \coloneqq 2\alpha\lambda_{\text{min}}(P)$, \cite[Thm. 2.1]{nesic-2008} ensures that there exists $\underline{T}^* \in \Re_{>0}$ such that, for all $\underline{T} \in (0, \underline{T}^*)$ and for all $x_0 \in \mathcal{F} \cup \mathcal{J}$, the solutions $x$ of HHB$(K)$ are such that $\tilde{x}$ converges to zero with an exponential rate of $a_3/(2a_2) = 2\alpha\lambda_{\text{min}}(P)/(2\lambda_{\text{max}}(P)) = \alpha/\text{cond}(P)$.
\end{proof}

The role of $M_{\varepsilon}$ in Theorem~\ref{thm:lmi_hhbm} motivates the results of the next section, where we take advantage of a similar approach in order to arrive at LMI conditions that are more numerically tractable and interpretable than \eqref{eq:lmi_hhb}. See \ref{sec:numerical_lmi} for more discussion on the feasibility of \eqref{eq:lmi_hhb}.

\section{Discrete-time Exponential Rates}

\subsection{Discrete-time dynamic equations}

Throughout this section, we assume that $\phi$ is a $\mu$-strongly convex function satisfying Assumption~\ref{objective_assumptions}. For certain choices of $\beta$, the following system, with state $x \coloneqq (q, p)$ and parameter $\epsilon \in \Re_{>0}$, can be viewed as a discretization of the systems studied in the previous sections:
\begin{align}
\label{eq:hhbm_dt}
\begin{split}
q_{k+1} &= q_k + \epsilon p_{k+1}, \\
p_{k+1} &= \beta(x_k) p_k - \epsilon \nabla\phi(q_k).
\end{split}
\end{align}
Setting $\beta \equiv 1 - \epsilon K$, the system \eqref{eq:hhbm_dt} is a discretization of HB($K$) defined in \eqref{eq:hb}, which we refer to as Polyak's method. On the other hand, defining $\beta$ to be
\begin{align}
\beta(x_k) &\coloneqq \beta(x_k; \underline{\beta}, \overline{\beta}) \nonumber \\
    &\coloneqq \begin{dcases}
        \overline{\beta} \coloneqq 1 - \epsilon\underline{K} & \text{if } \langle\nabla\phi(q_k), p_k\rangle < 0, \\
        \underline{\beta} \coloneqq 1 - \epsilon\overline{K} & \text{if } \langle\nabla\phi(q_k), p_k\rangle \geq 0, \\
    \end{dcases} \label{eq:beta_hihb}
\end{align}
\eqref{eq:hhbm_dt} is a discretization of HiHB($\underline{K}$, $\overline{K}$) and is referred to as HiHB-Pol with parameters $0 \leq \underline{\beta} \leq \overline{\beta} \leq 1$. For the case of $\underline{\beta} = 0$ and $\overline{\beta} = 1 - \epsilon K$ in \eqref{eq:beta_hihb}, the resulting system in \eqref{eq:hhbm_dt} is a discretization of HHB($K$) and is referred to as HHB-Pol with parameter $\overline{\beta}$. Thus, HHB-Pol is a special case of HiHB-Pol. 

Note that, when discretizing HHB$(K)$, there is no need to account for $\tau$ and $\underline{T}$. To see why, consider the case of $\underline{\beta} = 0$ and $\overline{\beta} = 1 - \epsilon K$ in \eqref{eq:beta_hihb}, and observe that, for any $k$ such that $\nabla\phi(q_k) \neq 0$, and for sufficiently small $\epsilon$, $\langle\nabla\phi(q_{k+1}), p_{k+1}\rangle = -\epsilon\nabla\phi(q_k)^T \nabla\phi(q_{k+1}) \simeq -\epsilon |\nabla\phi(q_k)| < 0$. Thus, $\beta(x_k) = \underline{\beta}$ implies that $\beta(x_{k+1}) = \overline{\beta}$. In other words, assuming sufficiently small $\epsilon$, whenever $x_k$ reaches a ``jump'' state, it immediately returns to a ``flow'' state and remains in ``flow'' states for some iterations thereafter. In this sense, HHB-Pol naturally models the temporally regularized system HHB$(K)$ by having trajectories in which each ``jump'' is followed by a period of ``flow''. 

\begin{remark}
The form \eqref{eq:hhbm_dt} has been referred to as a two-step or multi-step discretization \cite{muehlebach-2020}, which is related to symplectic integration \cite{hairer-2006}. From this viewpoint, other systems of interest in optimization can be obtained from \eqref{eq:hhbm_dt}. Setting $\underline{\beta} = 0$ and $\overline{\beta} = 1$ in \eqref{eq:beta_hihb}, \eqref{eq:hhbm_dt} becomes a discretization of the Hybrid Hamiltonian Algorithm of \cite{teel-2019}. Setting $\beta \equiv 1$, \eqref{eq:hhbm_dt} is a symplectic integration of Hamiltonian flow, i.e., the left-hand variant of \cite[Thm. 3.3]{hairer-2006}. Related systems are found in \cite[Eq. 26]{muehlebach-2020} and \cite[Sec. 2]{muehlebach-2019}.
\end{remark}

The results of the next section will be applicable to two different discretizations of HHBM, one based on Polyak's method and another based on Nesterov's method \cite{nesterov-2014}, which we now describe. These two discretizations will also be possible for HiHBM. First, the system \eqref{eq:hhbm_dt} can be rewritten to resemble Polyak's method in \cite{polyak-1964}:
\begin{subequations}
\label{eq:hhbm_polyak}
\begin{align}
q_{k+1} &= q_k + \epsilon \left[\beta(x_k) p_k - \epsilon \nabla\phi(q_k)\right], \label{eq:hhbm_polyak_q} \\
p_{k+1} &= \frac{q_{k+1} - q_k}{\epsilon}.
\end{align}
\end{subequations}
We have already described above how special cases of this system correspond to discretizations of HHBM and HiHBM.

Next, note that, for strongly convex objectives, the continuous-time limit of Polyak's method is the same differential equation as the continuous-time limit of Nesterov's method \cite{shi-2018}. Hence, we also consider the following system to be a discretization of our proposed hybrid and hybrid-inspired systems:
\begin{subequations}
\label{eq:hhbm_nesterov}
\begin{align}
q_{k+1} &= q_k + \epsilon \left[\beta(x_k) p_k - \epsilon \nabla\phi(q_k + \epsilon\beta(x_k) p_k)\right], \label{eq:hhbm_nesterov_q} \\
p_{k+1} &= \frac{q_{k+1} - q_k}{\epsilon}.
\end{align}
\end{subequations}
For the case $\beta \equiv 1 - \epsilon K$, we simply refer to \eqref{eq:hhbm_nesterov} as Nesterov's method with parameter $K \in \Re_{>0}$. On the other hand, defining $\beta$ as in \eqref{eq:beta_hihb}, the system \eqref{eq:hhbm_nesterov} is a discretization of HiHBM and is referred to as HiHB-Nes. Then, for the special case of $\underline{\beta} = 0$ and $\overline{\beta} = 1 - \epsilon K$ in \eqref{eq:beta_hihb}, the system \eqref{eq:hhbm_nesterov} is a discretization of HHBM, referred to as HHB-Nes.

In subsequent sections, it will be convenient to define the stepsize parameter $h \coloneqq \epsilon^2$, which appears as the coefficient of the gradient in the above discrete-time systems.

\subsection{LMI conditions}

Toward the goal of deriving LMI conditions for exponential convergence, consider the system
\begin{subequations}
\label{eq:lure}
\begin{align}
    x_{k+1} &= \hat{A}x_k + \hat{B}u_k, & y_k &= \hat{C}x_k, \\
    u_k &= \nabla\phi(y_k), & \xi_k &= \hat{E}x_k,
\end{align}
\end{subequations}
having a fixed point $(x^*, u^*, y^*, \xi^*)$ that satisfies $\xi^* = q^* \coloneqq \argmin_{q \in \mathbb R^n} \phi(q)$ and
\begin{subequations}
\begin{align*}
    x^* &= \hat{A}x^* + \hat{B}u^*, & y^* &\coloneqq \hat{C}x^*, \\
    u^* &\coloneqq \nabla\phi(y^*), & \xi^* &= \hat{E}x^*.
\end{align*}
\end{subequations}

We now begin rewriting \eqref{eq:hhbm_polyak_q} and \eqref{eq:hhbm_nesterov_q} in the form \eqref{eq:lure}. For the system  \eqref{eq:hhbm_polyak_q}, we define a set of system matrices in \eqref{eq:lure} for each case of the switching law \eqref{eq:beta_hihb}. Specifically, we set $x_k = (q_{k-1}, q_k)$ and, letting $h \coloneqq \epsilon^2$, we have
\begin{equation}
\begin{aligned}
\label{eq:lure_hb}
    \hat{A} &= \left[\begin{array}{cc}
        0 & I_n \\
        -\overline{\beta} I_n & (\overline{\beta} + 1) I_n
    \end{array}\right], &
    \hat{B} &= \left[\begin{array}{c}
        0 \\
        -h I_n
    \end{array}\right], \\
    \hat{C} &= \left[\begin{array}{c}
        0 \\
        I_n
    \end{array}\right]^T, &
    \hat{E} &= \left[\begin{array}{c}
        0 \\
        I_n
    \end{array}\right]^T.
\end{aligned}
\end{equation}
Then, define $(\hat{A}_R, \hat{B}_R, \hat{C}_R, \hat{E}_R)$ in the same way but with $\overline{\beta}$ replaced by $\underline{\beta}$. (In this case, $\hat{B} = \hat{B}_R$, $\hat{C} = \hat{C}_R$, and $\hat{E} = \hat{E}_R$.) The subscript $R$ indicates that these matrices represent the iterations that correspond to the (continuous-time) instants at which HHBM ``resets'' its $p$-state. 

For the system \eqref{eq:hhbm_nesterov_q}, define
\begin{align}
\label{eq:lure_nesterov}
        (\hat{A}, \hat{B}, \hat{E}) \text{ as in \eqref{eq:lure_hb},} \quad \hat{C} = \left[\begin{array}{c}
        -\overline{\beta} I_n \\
        (\overline{\beta} + 1) I_n
    \end{array}\right]^T.
\end{align}
Then, define $(\hat{A}_R, \hat{B}_R, \hat{C}_R, \hat{E}_R)$ in the same way but with $\overline{\beta}$ replaced by $\underline{\beta}$. (In this case, $\hat{B} = \hat{B}_R$, and $\hat{E} = \hat{E}_R$.)

For convenience, we use the following notation to distinguish between the ``non-reset'' and ``reset'' regions in the state-space of our proposed systems:
\begin{subequations}
\begin{align}
    S = &\{x = (x_1, x_2) \in \Re^{2n} : \nonumber \\ 
    &\quad \langle\nabla\phi(\hat{C}x), \; x_2 - x_1\rangle < 0\}, \label{eq:region_nonreset} \\
    S_R = &\{x = (x_1, x_2) \in \Re^{2n} : \nonumber \\ 
    &\quad \langle\nabla\phi(\hat{C}x), \; x_2 - x_1\rangle \geq 0\}. \label{eq:region_reset}
\end{align}
\end{subequations}
These sets reflect the switching law \eqref{eq:beta_hihb} but with $p_k$ scaled by $\epsilon$ (which does not change the nature of the law because $\epsilon > 0$). Then, combining the system matrices of \eqref{eq:lure_hb} or \eqref{eq:lure_nesterov} with \eqref{eq:lure}, we have a representation that can capture either system \eqref{eq:hhbm_polyak} or system \eqref{eq:hhbm_nesterov}, respectively, making both systems amenable to our LMI-based analysis:
\begin{subequations}
\label{eq:lure_hybrid}
\begin{align}
x_k &\in S \; \Longrightarrow & 
&\begin{dcases}
    x_{k+1} &= \hat{A}x_k + \hat{B}u_k, \\
    y_k &= \hat{C}x_k, \\
    u_k &= \nabla\phi(y_k), \\
    \xi_k &= \hat{E}x_k,
\end{dcases} \\
x_k &\in S_R \; \Longrightarrow & 
&\begin{dcases}
    x_{k+1} &= \hat{A}_Rx_k + \hat{B}_Ru_k, \\
    y_k &= \hat{C}_Rx_k, \\
    u_k &= \nabla\phi(y_k), \\
    \xi_k &= \hat{E}_Rx_k.
\end{dcases} 
\end{align}
\end{subequations}

We now have the ingredients to establish the following.
\begin{theorem}
\label{thm:lmi}
Let $\phi$ be a $\mu$-strongly convex function satisfying Assumption~\ref{objective_assumptions} with minimizer denoted $q^*$. With system matrices $(\hat{A}, \hat{B}, \hat{C}, \hat{E})$ given by either \eqref{eq:lure_hb} or \eqref{eq:lure_nesterov}, define the matrices
\begingroup
\allowdisplaybreaks
\begin{align}
    M_P &\coloneqq \left[\begin{array}{cc}
        \hat{A}^T P\hat{A} - \rho^2 P & \hat{A}^T P\hat{B} \\
        \hat{B}^T P\hat{A} & \hat{B}^T P\hat{B}
    \end{array}\right], \nonumber \\
    \Sigma_1 &\coloneqq \left[\begin{array}{cc}
        \hat{E}\hat{A} - \hat{C} & \hat{E}\hat{B} \\
        0 & I_n
    \end{array}\right], \quad \Sigma_2 \coloneqq \left[\begin{array}{cc}
        \hat{C} - \hat{E} & 0 \\
        0 & I_n
    \end{array}\right], \nonumber \\
    N_1 &\coloneqq \Sigma_1^T
    \left[\begin{array}{cc}
        \frac{L}{2}I_n & \frac{1}{2}I_n \\ [\smallskipamount]
        \frac{1}{2}I_n & 0
    \end{array}\right]
    \Sigma_1, \nonumber \\
    N_2 &\coloneqq \Sigma_2^T
    \left[\begin{array}{cc}
        \frac{-\mu}{2}I_n & \frac{1}{2}I_n \\ [\smallskipamount]
        \frac{1}{2}I_n & 0
    \end{array}\right]
    \Sigma_2, \nonumber \\
    N_3 &\coloneqq \left[\begin{array}{cc}
        \hat{C} & 0 \\
        0 & I_n
    \end{array}\right]^T
    \left[\begin{array}{cc}
        \frac{-\mu}{2}I_n & \frac{1}{2}I_n \\ [\smallskipamount]
        \frac{1}{2}I_n & 0
    \end{array}\right]
    \left[\begin{array}{cc}
        \hat{C} & 0 \\
        0 & I_n
    \end{array}\right], \nonumber \\
    M_1 &\coloneqq N_1 + N_2, \quad M_2 \coloneqq N_1 + N_3, \nonumber \\
    M_3 &= M_{\phi} \text{ as defined by \eqref{eq:M_phi},} \nonumber \\
    M &\coloneqq \left[\begin{array}{ccc}
        0 & 0 & \frac{1}{2} I_n \\ [\smallskipamount]
        0 & 0 & -\frac{1}{2} I_n \\ [\smallskipamount]
        \frac{1}{2} I_n & -\frac{1}{2} I_n & 0
    \end{array}\right]. \label{eq:M}
\end{align}%
\endgroup
Define $(M_{P,R}, M_{1,R}, M_{2,R}, M_{3,R})$ in the same way except with system matrices $(\hat{A}_R, \hat{B}_R, \hat{C}_R, \hat{E}_R)$. Suppose that there exist $a, \lambda, \lambda_R, \sigma, \sigma_R \in \Re_{>0}$, $\rho \in (0, 1]$, and a positive definite $P \in \Re^{2n \times 2n}$ such that
\begin{subequations}
\label{eq:lmi_system}
\begin{align}
    M_P &+ a\rho^2 M_1 + a(1 - \rho^2)M_2 \nonumber \\ 
    &+ \lambda M_3 + \sigma M \leq 0,  \label{eq:lmi} \\
    M_{P,R} &+ a\rho^2 M_{1,R} + a(1 - \rho^2)M_{2,R} \nonumber \\ 
    &+ \lambda_R M_{3,R} - \sigma_R M \leq 0. \label{eq:lmi_R}    
\end{align}
\end{subequations}
Then, there exists $c \in \Re_{>0}$ such that the trajectory of \eqref{eq:lure} satisfies
\begin{align}
    \phi(\xi_k) - \phi(\xi^*) \leq c \rho^{2k} \nonumber
\end{align}
for each initial condition $x_0 \in \Re^{2n}$ and for all $k \in \mathbb{Z}_{\geq 0}$. In particular,
\begin{align*}
    c = \frac{1}{a}\left(a(\phi(\xi_0) - \phi^*) + (x_0 - x^*)^T P (x_0 - x^*)\right).
\end{align*}
\end{theorem}
\begin{proof}
With $\xi \coloneqq \hat{E}x$, we will show that the function $V_k$ given by
\begin{align}
\label{eq:V_k}
\begin{split}
    P_k &\coloneqq \rho^{-2k}P, \quad a_k \coloneqq a\rho^{-2k}, \\
    V_k(x) &\coloneqq \rho^{-2k}\left(a(\phi(\xi) - \phi^*) + (x - x^*)^T P (x - x^*)\right)
\end{split}
\end{align}
satisfies
\begin{align}
\label{eq:V_nonincrease}
    V_{k+1}(x_{k+1}) \leq V_k(x_k),
\end{align}
for every iteration $k$ of \eqref{eq:lure_hybrid}, from which it follows that
\begin{align*}
    a_k(\phi(\xi_k) - \phi^*) \leq V_k(x_k) \leq V_0(x_0),
\end{align*}
and therefore,
\begin{align*}
    \phi(\xi_k) - \phi^* \leq \left(\frac{V_0(x_0)}{a}\right)\rho^{2k}.
\end{align*}
First, with $u$ defined by \eqref{eq:lure} and letting 
\begin{align}
\label{eq:e}
    e_k \coloneqq [x_k - x^* \quad u_k - u^*]^T,
\end{align}
we note that $M$ in \eqref{eq:M} satisfies
\begin{subequations}
\begin{align}
    e_k^T M e_k &\geq 0, \quad \forall k \text{ s.t. } x_k \in S, \label{eq:M_S} \\
    e_k^T M e_k &\leq 0, \quad \forall k \text{ s.t. } x_k \in S_R. \label{eq:M_S_R}
\end{align}
\end{subequations}
Now, consider any $k$ such that $x_k \in S$. Multiply \eqref{eq:lmi} by $e_k^T$ and $e_k$ on the left and right, respectively. Then, observe that \eqref{eq:M_S} implies
\begin{align}
\label{eq:sigma_M_S}
    \sigma e_k^T M e_k \geq 0,
\end{align}
while Lemma~\ref{thm:iqc} implies
\begin{align}
\label{eq:lambda_M}
    \lambda e_k^T M_3 e_k \geq 0,
\end{align}
and therefore,
\begin{align}
    e_k^T\left(M_P + a\rho^2 M_1 + a(1 - \rho^2)M_2\right)e_k \leq 0.
\end{align}
Letting
\begin{align*}
    M_{P_k} &\coloneqq \left[\begin{array}{cc}
        \hat{A}^T P_{k+1} \hat{A} - P_k & \hat{A}^T P_{k+1} \hat{B} \\
        \hat{B}^T P_{k+1} \hat{A} & \hat{B}^T P_{k+1} \hat{B}
    \end{array}\right],
\end{align*}
multiply the previous inequality by $\rho^{-2k-2}$ to obtain
\begin{align}
\label{eq:V_nonincrease_bound}
    e_k^T\left(M_{P_k} + a_k M_1 + (a_{k+1} - a_k)M_2\right)e_k \leq 0.
\end{align}
We will use \eqref{eq:V_nonincrease_bound} momentarily.

Next, from \cite[Lemma 4.1]{fazlyab-2018}, we have
\begin{align}
    \phi(\xi_{k+1}) - \phi(\xi_k) \leq e_k^T M_1 e_k, \label{eq:M1_bound} \\
    \phi(\xi_{k+1}) - \phi^* \leq e_k^T M_2 e_k. \label{eq:M2_bound} 
\end{align}
Multiply \eqref{eq:M1_bound} by $a_k$, multiply \eqref{eq:M2_bound} by $(a_{k+1} - a_k)$, and then add the resulting inequalities to obtain
\begin{align}
    &a_{k+1}(\phi(\xi_{k+1}) - \phi^*) - a_k(\phi(\xi_k) - \phi^*) \nonumber \\
    &\quad \leq e_k^T\left(a_k M_1 + (a_{k+1} - a_k)M_2\right)e_k.
\end{align}
Combine this inequality with the fact that
\begin{align*}
    e_k^T M_{P_k} e_k &= (x_{k+1} - x^*)^T P_{k+1} (x_{k+1} - x^*) \nonumber \\
    &\quad - (x_k - x^*)^T P_k (x_k - x^*)
\end{align*}
to obtain
\begin{align}
    &V_{k+1}(x_{k+1}) - V_k(x_k) \nonumber \\ 
    &\quad \leq e_k^T\left(M_{P_k} + a_k M_1 + (a_{k+1} - a_k)M_2\right)e_k.
\end{align}
Combining this inequality with \eqref{eq:V_nonincrease_bound}, we have shown that \eqref{eq:V_nonincrease} holds for all $k$ such that $x_k \in S$.

Finally, for any $k$ such that $x_k \in S_R$, we may show \eqref{eq:V_nonincrease} using the same steps above, replacing $(M_P, M_1, M_2, M_3)$ with $(M_{P,R}, M_{1,R}, M_{2,R}, M_{3,R})$ and replacing \eqref{eq:sigma_M_S} with
\begin{align*}
    -\sigma_R e_k^T M e_k \geq 0,
\end{align*}
which follows from \eqref{eq:M_S_R}. Because $S \cup S_R = \Re^{2n}$, we have shown that \eqref{eq:V_nonincrease} holds for all $k$.
\end{proof}

Theorem~\ref{thm:lmi} is a generalization of \cite[Thm 3.2]{fazlyab-2018}, extending the class of systems from those of the form \eqref{eq:lure} to those of the form \eqref{eq:lure_hybrid}.

\section{Numerical Results}
\label{sec:numerical}

\subsection{LMI solutions}
\label{sec:numerical_lmi}

We now demonstrate how Nesterov's method is impacted by uncertainty about the problem parameters $(\mu, L)$ and how HHB-Nes \eqref{eq:hhbm_nesterov} shows some promise for mitigating these effects. For simplicity, we use the notation $\beta$ in this section for both the $\beta$ parameter of Nesterov's method as well as the $\overline{\beta}$ parameter of HHB-Nes. We are interested in a setting in which the stepsize parameter $h$ and momentum parameter $\beta$ deviate significantly from the optimal values $(h^*, \beta^*)$, by which we mean the values dependent on $(\mu, L)$ that have been shown in, e.g., \cite[Proposition 12]{lessard-2016}, to be optimal with respect to the class of $\mu$-strongly convex objectives with $L$-Lipschitz gradient. Specifically, for both Nesterov's method and HHB-Nes, we set $h = 1/(2L)$ and $\beta = 1 - 0.1\sqrt{h}$, which makes $h$ an underestimate of $h^*$ and makes $\beta$ close to $1$ but still increasing with $L$.

To obtain a value of $\rho$ from the LMI \eqref{eq:lmi_system}, we perform a bisection search on $\rho$, solving the resulting LMI for each fixed value of $\rho \in [0, 1]$. From the discussion in \cite[Sec. 4.2]{lessard-2016}, the structure of the system matrices in \eqref{eq:lure_nesterov} ensures that the LMI \eqref{eq:lmi_system} holds with $n \geq 1$ being the dimension of $\text{dom }\phi$ if and only if it holds for $n=1$. So, we attempt to solve the LMI only for $n=1$ here. For each fixed $\rho$, the LMI is solved by SeDuMi $1.3$ in Matlab $2020a$. The resulting values of $\rho$ for $\mu = 1$ and $L \in [1, 100]$ are shown by the top two curves in Figure~\ref{fig:lmi_nesterov} in tuning $(h, \beta)$ causes $\rho$ to deteriorate significantly for Nesterov's method, while HHB-Nes mitigates these impacts. For Nesterov's method, we use the LMI in \cite[Thm 3.2]{fazlyab-2018} of which Theorem~\ref{thm:lmi} is an extension. The bottom two curves in Figure~\ref{fig:lmi_nesterov_hihbm} depict the case $(h, \beta) = (h^*, \beta^*)$, showing that, when the algorithmic parameters are tuned with perfect knowledge of $(\mu, L)$, HHBM does not necessarily preserve the scaling behavior of $\rho$ with respect to $L/\mu$, which is a desirable property of Nesterov's method \cite[Sec. 4.5]{lessard-2016}. The experiment is repeated for HiHB-Nes with $\underline{\beta} = 1 - \sqrt{h}$ in Figure~\ref{fig:lmi_nesterov_hihbm}, which shows that HiHB-Nes trades off between the behaviors of HHB-Nes and Nesterov's method.

We pursued similar experiments for HHB-Pol, in which the LMI \eqref{eq:lmi_system} was able to show that, for various choices (and sequences of choices) of algorithmic parameters, HHB-Pol achieves nearly the same convergence rate as Polyak's method across a variety of values of $L/\mu$. In particular, setting $\mu = 1$ and using parameters that are (locally) optimal for the strongly convex setting \cite[Sec. 4.2]{lessard-2016}, HHB-Pol achieves the same convergence rate as Polyak's method up to the value of $L$ at which rates can no longer be guaranteed for Polyak's method, reproducing the curve in \cite[Figure 5]{lessard-2016} labelled ``LMI (sector)''. These results suggest that HHB-Pol at least preserves the rates achievable by Polyak's method, even if it does not improve on those rates for any particular value of $L/\mu$.

The continuous-time LMI \eqref{eq:lmi_hhb} is difficult to solve (performing a bisection search for $\alpha \in [10^{-6}, 100]$) unless the $B = [0 \; -I_n]^T$ matrix is replaced by $[-I_n \; -I_n]^T$, in which case the LMI is feasible for a restrictive range of values for $L/\mu$, and the rates are difficult to compare with those of the heavy-ball differential equation. We leave it to future research to determine the permissible modifications to $(A, B, C, M_{\varepsilon})$ that could play a role in improving the feasibility of LMI conditions such as \eqref{eq:lmi_hhb}.

\begin{figure}
     \centering
     \begin{subfigure}[b]{0.45\textwidth}
         \centering
         \includegraphics[width=\linewidth]{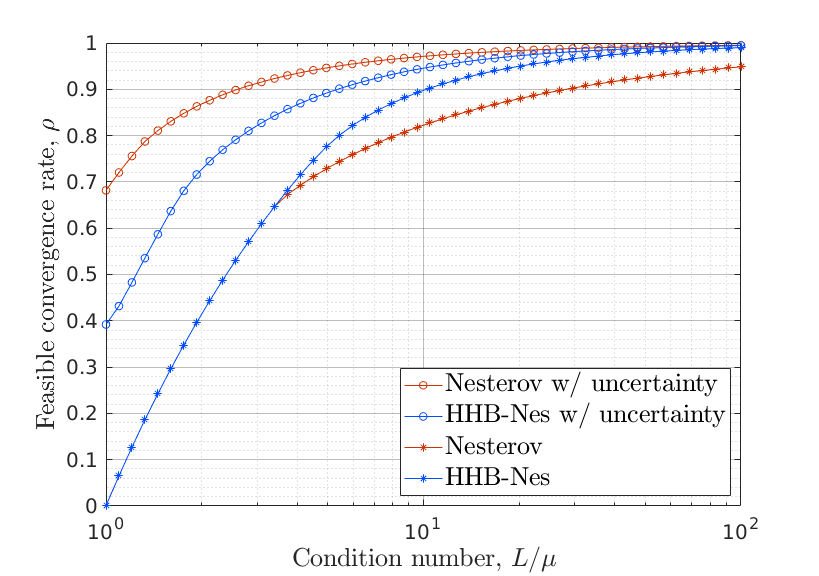}
         \caption{$\underline{\beta} = 0$}
         \label{fig:lmi_nesterov}
     \end{subfigure}
     \hfill
     \begin{subfigure}[b]{0.45\textwidth}
         \centering
         \includegraphics[width=\linewidth]{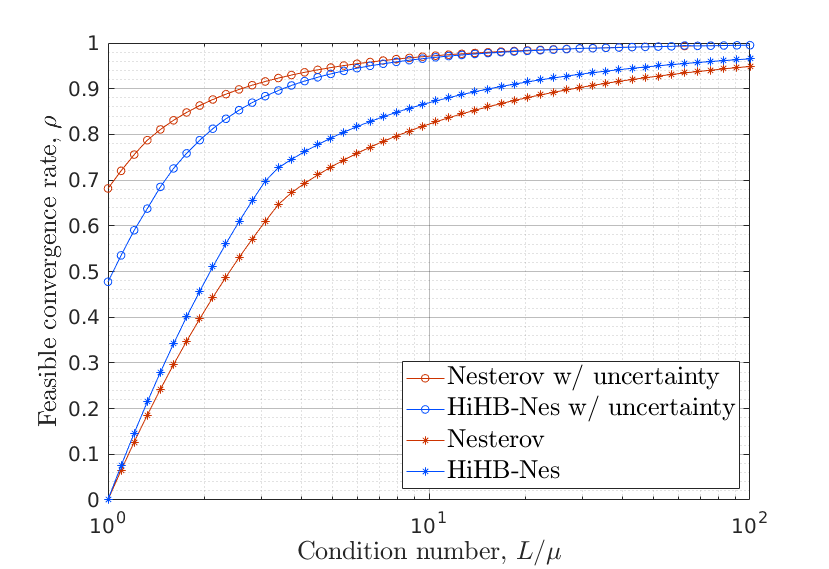}
         \caption{$\underline{\beta} = 1 - \sqrt{h}$}
         \label{fig:lmi_nesterov_hihbm}
     \end{subfigure}
        \caption{Values of $\rho$ obtained from Thm.~\ref{thm:lmi} for HHB-Nes and HiHB-Nes and from \cite[Thm 3.2]{fazlyab-2018} for Nesterov's method. The circled curves indicate uncertainty about $(\mu, L)$ in tuning $(h, \beta)$, as described in Sec.~\ref{sec:numerical_lmi}. The starred curves indicate that $h$ and $\beta$ are tuned with perfect knowledge of $(\mu, L)$, using \cite[Proposition 12]{lessard-2016}.}
        \label{fig:lmi}
\end{figure}

\subsection{Example: strongly convex quadratic objective}

In this section, we show examples of how HHBM can prevent oscillations from appearing in the trajectory of $\phi(q)$ that are caused by having a poorly tuned pair $(h, \beta)$. In particular, we consider a problem with strongly convex quadratic objective $\phi(q) \coloneqq \frac{1}{2}q^T Q q + b^T q$, with $L/\mu = 10^3$, and with $Q$ and $b$ randomly generated as described at the end of this subsection. As in the previous section, we intentionally use a smaller value of $h$ than recommended in theory ($h = 10^{-4}$), forcing $\beta = 1 - \sqrt{h}K$ to take values very close to $1$. As a consequence, it is more intuitive to discuss the parameters of the algorithms in terms of $K$ rather than $\beta$, which we do in Figure~\ref{fig:phi_dt}, where $K = 1.97$ roughly corresponds to the value that yields the fastest convergence rate for Polyak's method and Nesterov's method for the given $h$ and given $(Q, b)$ (determined experimentally). Figure~\ref{fig:phi_dt} shows that the convergence rates of both Polyak's method and Nesterov's method deteriorate significantly when $K$ underestimates the desirable value of $1.97$. Furthermore, due to the large $L/\mu$, the only way to remove the oscillations from the trajectory of Polyak's method is to increase $K$ to the point at which the asymptotic convergence rate is significantly slower than seen in Figure~\ref{fig:phi_dt_polyak}. In contrast, HHBM achieves the same asymptotic rates as the other two methods when $K = 1.97$, while it exhibits both faster asymptotic rates and fewer oscillations when $K$ underestimates the desirable value. 

We do not include HiHBM in Figure~\ref{fig:phi_dt} because the chosen stepsize is sufficiently small that the performance of HiHBM can be made to resemble that of HHBM very closely by choosing $\overline{K}$ sufficiently large. The appeal of HiHBM will instead be conveyed in the next section.

It is important to note that the advantages of HHBM and HiHBM are specific to the situation simulated above, in which the optimal algorithmic parameters are unavailable, especially when $h$ is too small and $\beta$ is too large. In our experiments, we have found situations where Polyak's method and Nesterov's method do not exhibit oscillations when using their respective optimally tuned parameters, and in these situations, our proposed algorithms essentially achieve the same asymptotic convergence rates (and lack of oscillations) as their classic counterparts. These observations are compatible with our LMI computations of the previous section, which suggested that our algorithms do not improve on the rates of their classic counterparts when optimal algorithmic parameters are available.

We use the following approach to generate a random matrix with a specific condition number $L/\mu$. In fact, we only consider $\mu = 1$. First, generate a random $n \times n$ matrix. Then, take the singular value decomposition $USV^T$, and replace $S$ with $\hat S$, where $\hat S$ has diagonal entries $\sigma_i$ for $i \in \{1, \ldots, n\}$ satisfying $\sigma_{\min} = 1$ and $\sigma_{\max} = \sqrt{L}$, with each of the remaining $n-2$ diagonal entries being uniformly distributed on $[1, \sqrt{L}]$. Let $\hat Q = U\hat SV^T$, and take $Q \coloneqq \hat Q \hat Q^T$ to be the matrix that defines $\phi$. To generate $b$, we take each of its entries to be uniformly distributed on $[-100, 100]$. Initial conditions are randomly generated in the same way as $b$. We have verified that the behaviors in Figure~\ref{fig:phi_dt} persist across several random trials.

\begin{figure}
     \centering
     \begin{subfigure}[b]{0.45\textwidth}
         \centering
         \includegraphics[width=\linewidth]{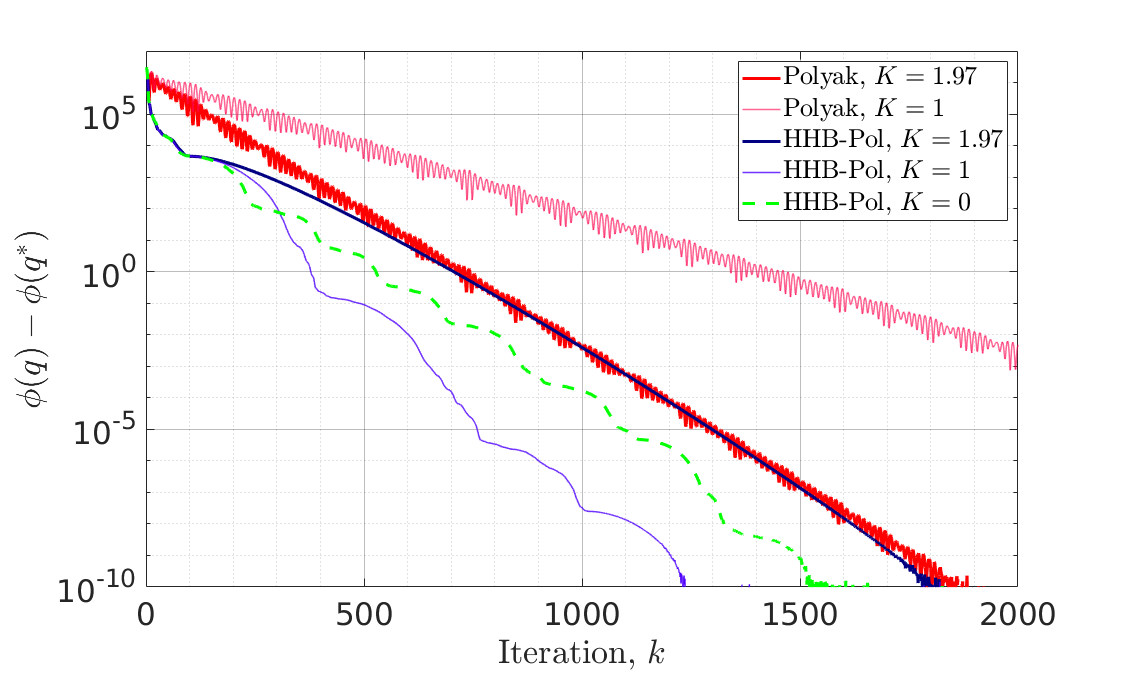}
         \caption{}
         \label{fig:phi_dt_polyak}
     \end{subfigure}
     \hfill
     \begin{subfigure}[b]{0.45\textwidth}
         \centering
         \includegraphics[width=\linewidth]{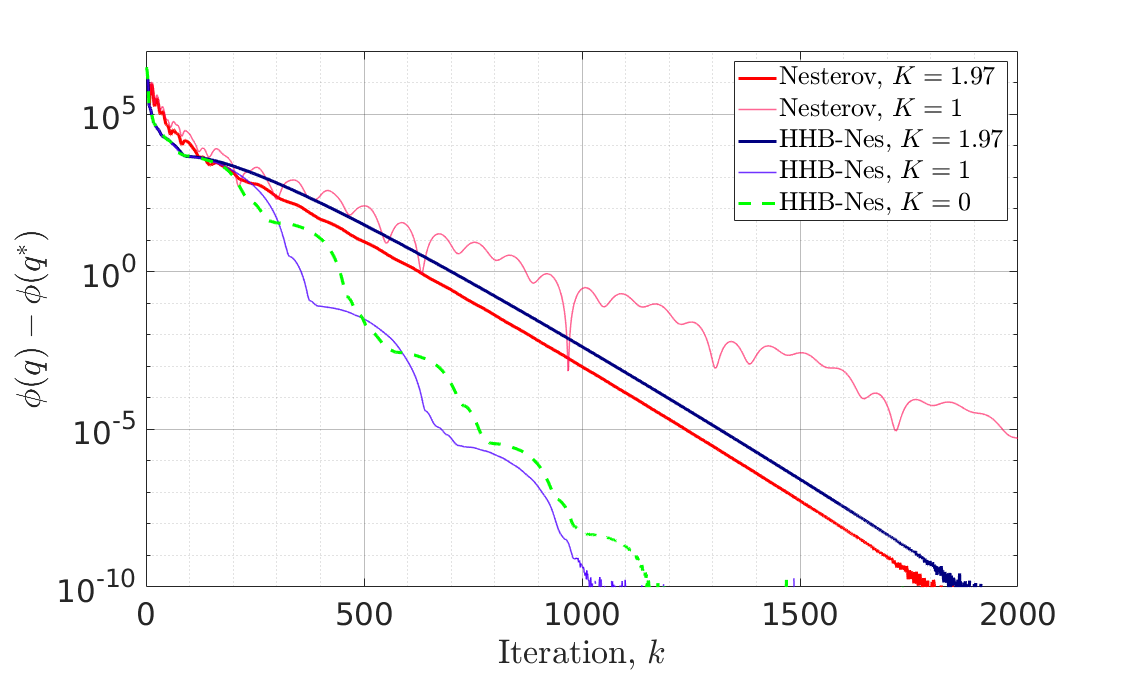}
         \caption{}
         \label{fig:phi_dt_nesterov}
     \end{subfigure}
        \caption{The value of $\phi(q) - \phi(q^*)$ versus iteration $k$, where $\phi$ is a strongly convex quadratic objective, and $q^*$ is the minimizer of $\phi$ computed by Matlab's ``quadprog'' function.}
        \label{fig:phi_dt}
\end{figure}

\subsection{Example: logistic regression}
\label{sec:lr}
In this section, we show that our proposed algorithms may have advantages even when applied to problems that are more general than those considered in our analyses. We consider a problem of logistic regression, common in statistical data analysis and machine learning \cite[Ch. 4]{hastie-2001}. Figure \ref{fig:lr} compares convergence rates for a problem of logistic regression with dataset $(\Theta, b)$, where $\Theta \in \Re^{n \times m}$ has columns denoted $\Theta_i$ and entries randomly drawn with standard normal distribution, and each component $b_i$ of $b$ is uniformly randomly drawn from $\{-1, 1\}$. The dataset has $m = 1000$ observations. The objective is
\begin{align}
\label{eq:lr_objective}
\phi(q) &= \sum_{i = 1}^{m} \log\left( 1 + \exp\left(-b_i \Theta_i^T q\right)\right).
\end{align}
The objective is convex but not strictly convex. However, on any compact set, it satisfies the $\mu$-PL condition \cite[Eq. 3]{karimi-2016} for some constant $\mu$ (see \cite[Sec. 2.3]{karimi-2016} for a discussion).

The stepsizes of gradient descent, Nesterov's method, and Polyak's method were tuned via bisection search. For Polyak's method, the best stepsize was chosen from a broad range of values, with the value of $\beta$ being tuned via bisection search for each stepsize considered. For Nesterov's method, we use the standard sequence of values for $\beta$ intended for the class of general convex objectives \cite[Sec. 2.2]{nesterov-2014}, namely the sequence $\alpha_k(1 - \alpha_k)/\left[\alpha_k^2 + \alpha_{k+1}\right]$ where $\alpha_k$ satisfies $\alpha_{k+1}^2 = (1 - \alpha_{k+1})\alpha_k^2$ (here, the initialization of $\beta$ had negligible effect). For HiHBM, we set $\underline{K} = 0$, while the stepsize and $\overline{K}$ were tuned in the same way that the stepsize and $\beta$ were tuned for HBM (which resulted in the same $\epsilon$ found for HBM). For HHBM, the stepsize and $\underline{K}$ were tuned in the same way that the stepsize and $\overline{K}$ were tuned for HiHBM, resulting in the same stepsize found for Nesterov's method and $\underline{K} = 0$.


\begin{figure}
    \centering
    \includegraphics[width=1\linewidth]{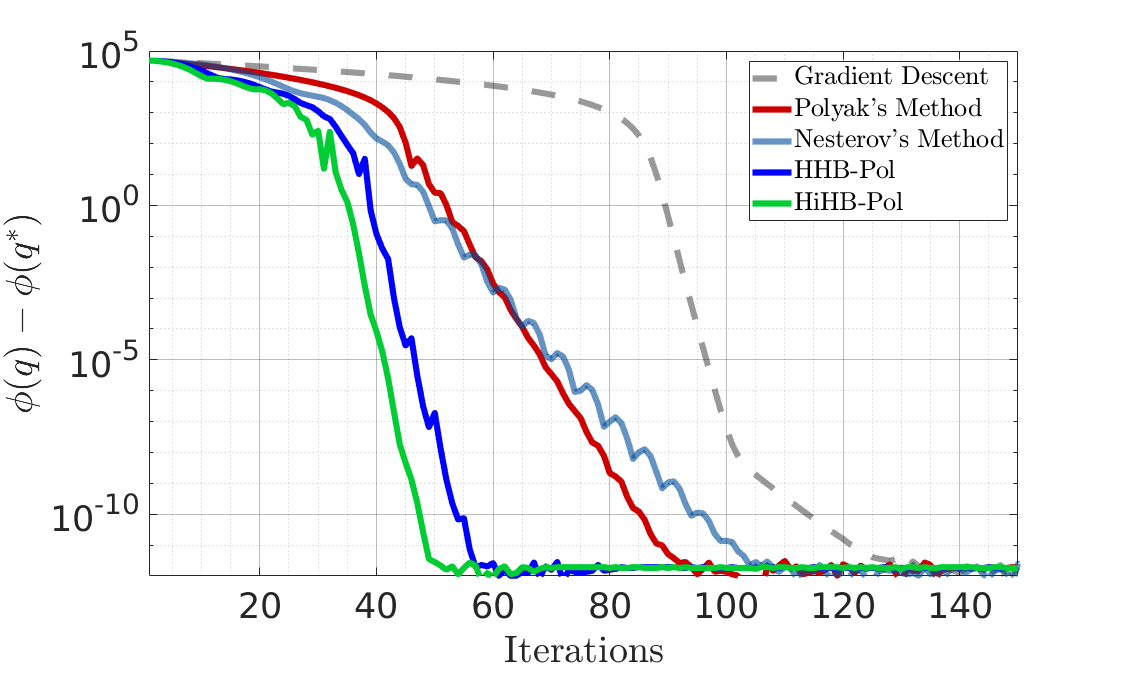}
    \caption{The value of $\phi(q) - \phi(q^*)$ versus iteration $k$, where $\phi$ is the logistic regression objective \eqref{eq:lr_objective}, and $q^*$ is the minimizer of $\phi$ computed by gradient descent.}
    \label{fig:lr}
\end{figure}

\ifx
\begin{figure}
     \centering
     \begin{subfigure}[b]{0.45\textwidth}
         \centering
         \includegraphics[width=\linewidth]{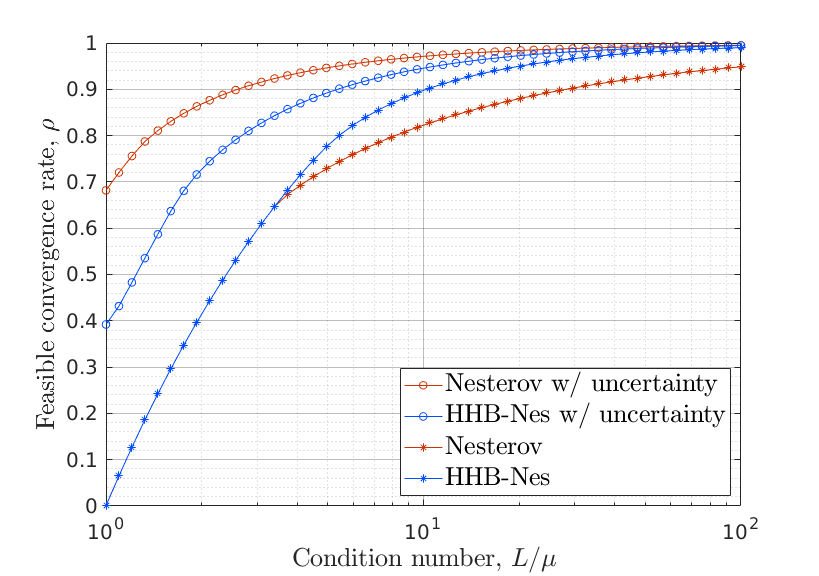}
         \caption{$\underline{\beta} = 0$}
         \label{fig:lmi_nesterov}
     \end{subfigure}
     \hfill
     \begin{subfigure}[b]{0.45\textwidth}
         \centering
         \includegraphics[width=\linewidth]{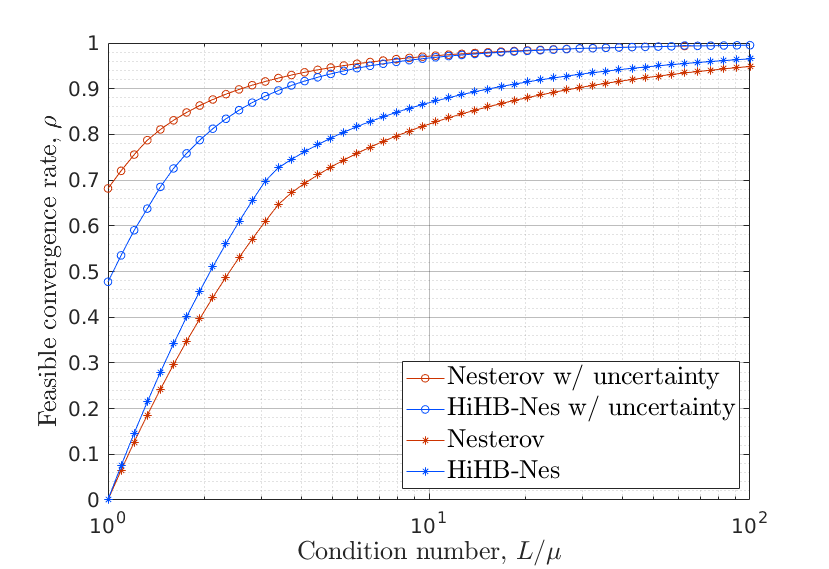}
         \caption{$\underline{\beta} = 1 - \sqrt{h}$}
         \label{fig:lmi_nesterov_hihbm}
     \end{subfigure}
        \caption{Values of $\rho$ obtained from Thm.~\ref{thm:lmi} for HHB-Nes and HiHB-Nes and from \cite[Thm 3.2]{fazlyab-2018} for Nesterov's method. The circled curves indicate uncertainty about $(\mu, L)$ in tuning $(h, \beta)$, as described in Sec.~\ref{sec:numerical_lmi}. The starred curves indicate that $h$ and $\beta$ are tuned with perfect knowledge of $(\mu, L)$, using \cite[Proposition 12]{lessard-2016}.}
        \label{fig:lmi}
\end{figure}

\begin{figure}[t!]
     \centering
     \begin{subfigure}[b]{0.24\textwidth}
         \centering
         \includegraphics[width=\linewidth]{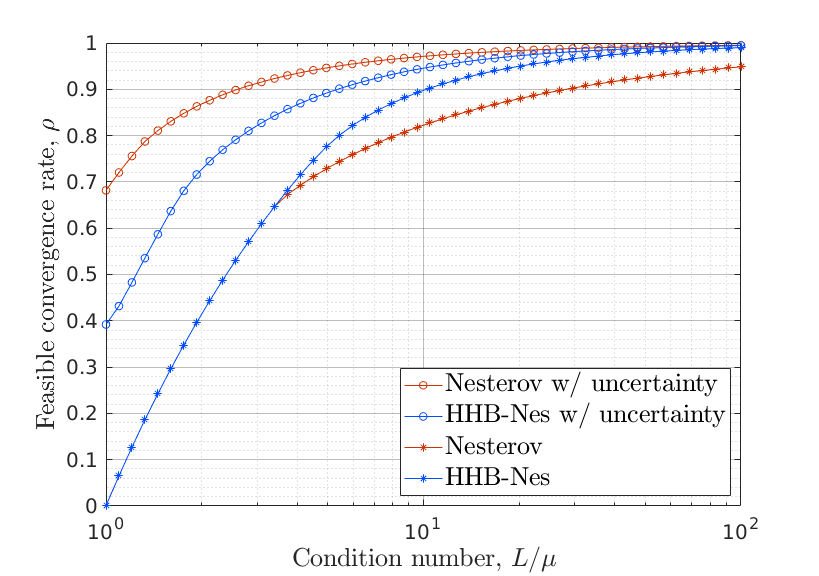}
         \caption{$\underline{\beta} = 0$}
         \label{fig:lmi_nesterov}
     \end{subfigure}%
     ~
     \begin{subfigure}[b]{0.24\textwidth}
         \centering
         \includegraphics[width=\linewidth]{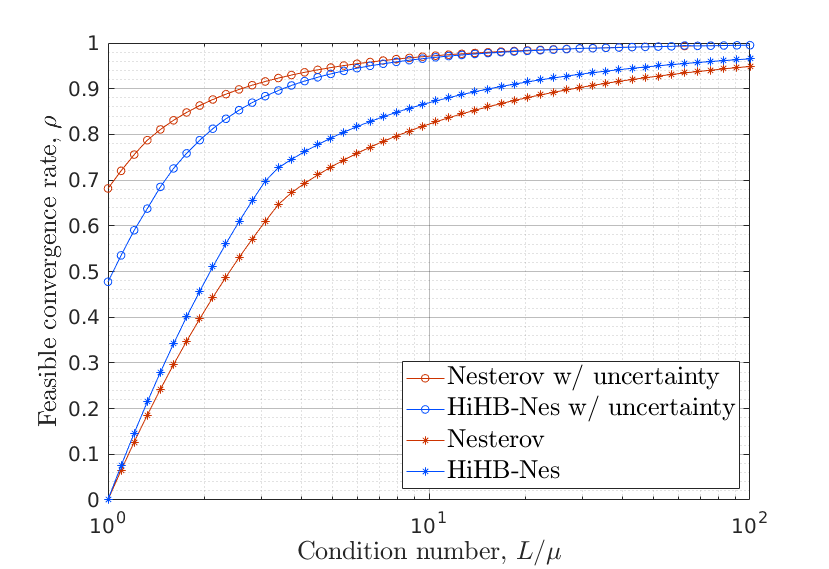}
         \caption{$\underline{\beta} = 1 - \sqrt{h}$}
         \label{fig:lmi_nesterov_hihbm}
     \end{subfigure}
        \caption{Values of $\rho$ obtained from Thm.~\ref{thm:lmi} for HHB-Nes and HiHB-Nes and from \cite[Thm 3.2]{fazlyab-2018} for Nesterov's method. The circled curves indicate uncertainty about $(\mu, L)$ in tuning $(h, \beta)$, as described in Sec.~\ref{sec:numerical_lmi}. The starred curves indicate that $h$ and $\beta$ are tuned with perfect knowledge of $(\mu, L)$, using \cite[Proposition 12]{lessard-2016}.}
        \label{fig:lmi}
\end{figure}
\fi

\bibliographystyle{IEEEtran}
\bibliography{main.bib}

\begin{thebibliography}{10}
\providecommand{\url}[1]{#1}
\csname url@rmstyle\endcsname
\providecommand{\newblock}{\relax}
\providecommand{\bibinfo}[2]{#2}
\providecommand\BIBentrySTDinterwordspacing{\spaceskip=0pt\relax}
\providecommand\BIBentryALTinterwordstretchfactor{4}
\providecommand\BIBentryALTinterwordspacing{\spaceskip=\fontdimen2\font plus
\BIBentryALTinterwordstretchfactor\fontdimen3\font minus
  \fontdimen4\font\relax}
\providecommand\BIBforeignlanguage[2]{{%
\expandafter\ifx\csname l@#1\endcsname\relax
\typeout{** WARNING: IEEEtran.bst: No hyphenation pattern has been}%
\typeout{** loaded for the language `#1'. Using the pattern for}%
\typeout{** the default language instead.}%
\else
\language=\csname l@#1\endcsname
\fi
#2}}

\bibitem{nesterov-2014}
Y.~Nesterov, \emph{Introductory Lectures on Convex Optimization: A Basic
  Course}, 1st~ed.\hskip 1em plus 0.5em minus 0.4em\relax Springer Publishing
  Company, Incorporated, 2014.

\bibitem{lessard-2016}
L.~Lessard, B.~Recht, and A.~Packard, ``Analysis and design of optimization
  algorithms via integral quadratic constraints,'' \emph{SIAM Journal on
  Optimization}, vol.~26, no.~1, pp. 57--95, 2016.

\bibitem{o'donoghue-2015}
B.~{O}'{D}onoghue and E.~Cand\`{e}s, ``Adaptive restart for accelerated
  gradient schemes,'' \emph{Found. Comput. Math.}, vol.~15, no.~3, pp.
  715--732, June 2015.

\bibitem{hauswirth-2020}
A.~{Hauswirth}, S.~{Bolognani}, G.~{Hug}, and F.~{D\"{o}rfler}, ``Timescale
  separation in autonomous optimization,'' \emph{IEEE Transactions on Automatic
  Control}, pp. 1--1, 2020.

\bibitem{ortmann-2020}
L.~Ortmann, A.~Hauswirth, I.~Caduff, F.~Dörfler, and S.~Bolognani,
  ``Experimental validation of feedback optimization in power distribution
  grids,'' \emph{Electric Power Systems Research}, vol. 189, p. 106782, 2020.

\bibitem{fercoq-2019}
\BIBentryALTinterwordspacing
O.~Fercoq and Z.~Qu, ``Adaptive restart of accelerated gradient methods under
  local quadratic growth condition,'' \emph{{IMA} Journal of Numerical
  Analysis}, vol.~39, no.~4, pp. 2069--2095, Mar. 2019. [Online]. Available:
  \url{https://doi.org/10.1093/imanum/drz007}
\BIBentrySTDinterwordspacing

\bibitem{roulet-2017}
V.~Roulet and A.~d'Aspremont, ``Sharpness, restart and acceleration,'' in
  \emph{Advances in Neural Information Processing Systems 30}, I.~Guyon, U.~V.
  Luxburg, S.~Bengio, H.~Wallach, R.~Fergus, S.~Vishwanathan, and R.~Garnett,
  Eds.\hskip 1em plus 0.5em minus 0.4em\relax Curran Associates, Inc., 2017,
  pp. 1119--1129.

\bibitem{fercoq-2019b}
O.~Fercoq and Z.~Qu, ``Restarting the accelerated coordinate descent method
  with a rough strong convexity estimate,'' \emph{Computational Optimization
  and Applications}, vol.~75, no.~1, pp. 63--91, Oct. 2019.

\bibitem{prieur-2018}
\BIBentryALTinterwordspacing
C.~Prieur, I.~Queinnec, S.~Tarbouriech, and L.~Zaccarian, ``Analysis and
  synthesis of reset control systems,'' \emph{Foundations and Trends® in
  Systems and Control}, vol.~6, no. 2-3, pp. 117--338, 2018. [Online].
  Available: \url{http://dx.doi.org/10.1561/2600000017}
\BIBentrySTDinterwordspacing

\bibitem{teel-2019}
A.~R. {Teel}, J.~I. {Poveda}, and J.~{Le}, ``First-order optimization
  algorithms with resets and {H}amiltonian flows,'' in \emph{2019 IEEE 58th
  Conference on Decision and Control (CDC)}, 2019, pp. 5838--5843.

\bibitem{nesic-2008}
\BIBentryALTinterwordspacing
D.~Ne\v{s}i\'{c}, L.~Zaccarian, and A.~R. Teel, ``Stability properties of reset
  systems,'' \emph{Automatica}, vol.~44, no.~8, p. 2019–2026, Aug. 2008.
  [Online]. Available: \url{https://doi.org/10.1016/j.automatica.2007.11.014}
\BIBentrySTDinterwordspacing

\bibitem{ben-israel-1986}
A.~Ben-Israel and B.~Mond, ``What is invexity?'' \emph{The Journal of the
  Australian Mathematical Society. Series B. Applied Mathematics}, vol.~28,
  no.~1, p. 1–9, 1986.

\bibitem{goebel-2012}
R.~Goebel, R.~Sanfelice, and A.~Teel, \emph{Hybrid Dynamical Systems: Modeling,
  Stability, and Robustness}.\hskip 1em plus 0.5em minus 0.4em\relax Princeton
  University Press, 2012.

\bibitem{fazlyab-2018}
\BIBentryALTinterwordspacing
M.~Fazlyab, A.~Ribeiro, M.~Morari, and V.~M. Preciado, ``Analysis of
  optimization algorithms via integral quadratic constraints: Nonstrongly
  convex problems,'' \emph{{SIAM} J. Optimization}, vol.~28, no.~3, pp.
  2654--2689, 2018. [Online]. Available:
  \url{https://doi.org/10.1137/17M1136845}
\BIBentrySTDinterwordspacing

\bibitem{muehlebach-2020}
M.~Muehlebach and M.~I. Jordan, ``Optimization with momentum: Dynamical,
  control-theoretic, and symplectic perspectives,'' 2020,
  \url{https://arxiv.org/abs/2002.12493}.

\bibitem{hairer-2006}
\BIBentryALTinterwordspacing
E.~Hairer, C.~Lubich, and G.~Wanner, \emph{{Geometric Numerical Integration:
  Structure-Preserving Algorithms for Ordinary Differential Equations; 2nd
  ed.}}\hskip 1em plus 0.5em minus 0.4em\relax Dordrecht: Springer, 2006.
  [Online]. Available: \url{https://cds.cern.ch/record/1250576}
\BIBentrySTDinterwordspacing

\bibitem{muehlebach-2019}
M.~Muehlebach and M.~Jordan, ``A dynamical systems perspective on {N}esterov
  acceleration,'' in \emph{Proceedings of the 36th International Conference on
  Machine Learning}, ser. Proceedings of Machine Learning Research,
  K.~Chaudhuri and R.~Salakhutdinov, Eds., vol.~97.\hskip 1em plus 0.5em minus
  0.4em\relax Long Beach, California, USA: PMLR, 09--15 Jun 2019, pp.
  4656--4662.

\bibitem{polyak-1964}
B.~Polyak, ``Some methods of speeding up the convergence of iteration
  methods,'' \emph{{USSR} Computational Mathematics and Mathematical Physics},
  vol.~4, no.~5, pp. 1--17, Jan. 1964.

\bibitem{shi-2018}
B.~Shi, S.~S. Du, M.~I. Jordan, and W.~J. Su, ``Understanding the acceleration
  phenomenon via high-resolution differential equations,''
  \url{https://arxiv.org/abs/1810.08907}, 2018.

\bibitem{hastie-2001}
T.~Hastie, R.~Tibshirani, and J.~Friedman, \emph{The Elements of Statistical
  Learning}, ser. Springer Series in Statistics.\hskip 1em plus 0.5em minus
  0.4em\relax New York, NY, USA: Springer New York Inc., 2001.

\bibitem{karimi-2016}
H.~Karimi, J.~Nutini, and M.~Schmidt, ``Linear convergence of gradient and
  proximal-gradient methods under the {P}olyak-{L}ojasiewicz condition,'' in
  \emph{Machine Learning and Knowledge Discovery in Databases}, P.~Frasconi,
  N.~Landwehr, G.~Manco, and J.~Vreeken, Eds.\hskip 1em plus 0.5em minus
  0.4em\relax Cham: Springer International Publishing, 2016, pp. 795--811.

\end{thebibliography}
\end{document}